\documentclass[journal]{acm_proc_article-sp}

\usepackage{graphicx,algorithm,algpseudocode,algorithmicx}
\usepackage{amsmath}
\usepackage{amssymb}
\usepackage{mathtools}
\usepackage{xspace}
\usepackage{xcolor}

\newtheorem{theorem}{Theorem}
\newtheorem{lemma}{Lemma}

\newtheorem{definition}{Definition}
\newtheorem{example}{Example}

\def\S{\ensuremath{\mathcal{S}}\xspace}
\def\Sp{\ensuremath{\mathcal{S}^\prime}\xspace}
\def\reopt1{Reopt-CSP\ensuremath{_{\mathcal{M}_+}}\xspace}

\title{On Self-Reducibility and Reoptimization of Closest Substring Problem}

\numberofauthors{1}
\author{ Jeffrey A. Aborot,   Henry Adorna,  Jhoirene B. Clemente		\\
\affaddr{Algorithms and Complexity Laboratory}\\
		\affaddr{Department of Computer Science}\\
		\affaddr{College of Engineering}\\
		\affaddr{University of the Philippines Diliman}\\
		\email{jeffrey.aborot@up.edu.ph, ha@dcs.upd.edu.ph, jbclemente@up.edu.ph}
}

\begin{document}
\maketitle

\pagestyle{empty}
\thispagestyle{empty}

\begin{abstract}
In this paper, we define  the  reoptimization variant of the closest substring problem (CSP) under sequence addition. We show that, even with the additional information we have about the problem instance, the problem of finding a closest substring is still NP-hard. %We show that with reoptimization, characterized by adding $k$ sequences to the original instance, we can improve the running time of the existing PTAS with the same approximation ratio.
We investigate the combinatorial property of optimization problems called self-reducibility. We show that problems that are polynomial-time reducible to  self-reducible problems also exhibits the same property. We illustrate this in the context of CSP. We used the property to show that although we cannot improve the approximability of the problem, we can improve the running time of the existing PTAS for CSP.

%make use of any known $\sigma$-approximation algorithm  $ALG$ to obtain an improved approximation ratio of $(2\sigma-1)/ \sigma$ that runs in $O(t^k \cdot Time($ALG$) + (tn)^k)$, where $t$ is the total number of input sequences of length $n$.
\end{abstract}


%\IAENGpeerreviewmaketitle

\section{Introduction}
The \textit{consensus pattern problem} is a combinatorial problem applied to a wide range of applications from string matching problems in genomic sequences \cite{Bodlaender1995} to finding repeated patterns in graphs \cite{Dondi2013} and time-series databases \cite{Li2010}. Due to its utmost importance in the field of genomics, several efforts have been made to characterize the computational requirements needed to solve the problem. If the set of input instances is constrained only to a set of input strings, the problem is known as the \textit{closest substring problem} (CSP). The problem seeks to identify a pattern that occurs approximately in each of the given set of sequences. 

It has been shown that the CSP is NP-hard \cite{Garey1979}, i.e., unless P=NP, we cannot obtain a polynomial-time algorithm solving the problem. Moreover, it was also shown that  fixing  parameters such as pattern length or alphabet size, does not address the intractability of the problem \cite{Evans2003}. Aside from parameterization attempts, other studies tried to relax the condition of always finding the optimal solution by providing approximation algorithms for the problem. The first constant-factor approximation algorithm is presented in \cite{Lanctot1998}, then subsequently improved in  \cite{Li1999}. These results include CSP to the class of problems that are constant-factor approximable (APX). In addition to this, several studies  \cite{Li1999,Ma2000,Ma2008} even presented a \textit{polynomial-time approximation scheme} (PTAS) for the problem.

To address the intractability and to improve the solution quality of approximation algorithms, another approach is to use additional information about the problem instance whenever possible. A method called \textit{reoptimization} has already been applied to a variety of hard problems in the literature. The idea of reoptimization is to make use of a solution to a locally modified version of the input instance. It was shown that reoptimization can help to either improve the approximability and even provide a PTAS for some problems that are APX-hard. These results include improvements for the metric-traveling salesman problem \cite{Bockenhauer2008}, the Steiner tree problem \cite{Bilo2012}, the common superstring problem  \cite{Bilo2011}, and hereditary graph problems \cite{Boria2012}.

In this paper, we investigate whether reoptimization can help in approximating the closest substring problem. The additional information given in advance is a solution to a smaller instance of the problem. The algorithm for the reoptimization variant of the problem aims to make use of the given solution to become feasible for the larger instance. We show that, using the given solution as a greedy partial solution to the larger instance, we can achieve an additive error (with respect to the optimal solution) that grows linearly as we increase the number of sequences added to the original instance, which can be worse than the existing $\sigma$-approximation algorithms for the original problem. 

The self-reducibility property of some hard combinatorial problem has been showed to improve the approximability of the problem. However, existing general approaches incur a much longer time for providing a solution with improved quality. Providing an improved ratio is already possible for CSP, due to the existence of  a PTAS.   However, \cite{Wang2008} deemed the PTASs in the literature as impractical for small error bounds. In this paper, we showed that it is possible to improve the running time of the existing PTAS  for CSP while maintaining the same approximation ratio through reoptimization.

\section{Closest Substring Problem}
Given a set of sequences $\mathcal{S} = \{S_1, S_2, \ldots ,S_t \}$ defined over some alphabet $\Sigma$, where each $S_i$ is of length $n$, and for some $l \leq n$, find a string $v$ of length $l$ and a set of substrings $y_i$ each from $S_i$, where $v, y_i \in \Sigma^l$  such that the total Hamming distance $\sum\limits_{i}^t d(v, y_i)$ is minimized \cite{Garey1979}.

The closest substring of a given set of input sequences, may not be unique given the closest substring  occurrences $y_i$, for $ 1 \leq i \leq t$.  So instead of focusing on the closest substring $v$, we will use the collection of $y_i$'s to represent a feasible solution for the problem. We can represent a feasible solution $$SOL = (y_1, y_2, \ldots, y_t)$$ using a sequence of substrings obtained from each of the given sequences to \S. We may refer to SOL as a feasible solution and OPT as the optimal solution for input instance \S.

We can easily obtain a substring $v$ from the collection of occurrences SOL. This is done by aligning all substrings in SOL and taking the majority symbol for each column of the alignment. Let us call this string the \textit{consensus} of SOL. Moreover, the consensus of $OPT$ is called the closest substring of \S. It is easy to see that the consensus of OPT will minimize the total Hamming distance for all  $v \in \Sigma^l$. On the other hand, given a string $v \in \Sigma^l$, we can also obtain a set of occurrences SOL by searching the best aligning substrings from each of the sequences in \S. Given the closest substring, we can obtain OPT in $O(tn)$. 

Naively, we can search for the optimal solution by considering all possible locations of the closest substring in each sequence $S_i$ from $\mathcal{S}$. For smaller alphabet sizes, it is easier to search over all $\Sigma^l$ as compared to $O(n^t)$ possible occurrences. However, since we are considering the general CSP, we do not restrict the alphabet size of the given sequences.

\subsection{Related Works}
Algorithms for the CSP can be categorized into three classes: \textit{exact}, \textit{heuristics}, and \textit{approximate}. Exact algorithms always obtain the optimal solution, but are exhaustive and impractical due to the NP-hardness of the problem. Among the earliest exact algorithms, some use graphs to model the problem. This includes WINNOWER \cite{Pevzner2000} which involves finding cliques in graphs obtained from the set of sequences. An integer linear programming (ILP) formulation of the problem was also presented in \cite{Zaslavsky2006}.

Some algorithms make use of some strategy in searching through the set of all feasible solutions. Algorithms following this approach  are called \textit{heuristic algorithms}. A majority of these results use probabilistic models to represent solutions. %The earliest works include  Expectation Maximization (EM) \cite{Bailey1994} and Gibbs Sampling \cite{Roth1998}, where both are local search algorithms based on some preliminary statistical model. The problem with these approaches is that the local solutions obtained may not always be optimal. A result in \cite{Buhler2002} improved EM by adding a preliminary step called \textit{projection}. It provides an intelligent guess as to where the local search will start. Projection was subsequently improved in \cite{Alam2003} by replacing the preliminary step with \textit{aggregation}, a step which allows the algorithm to run twice as fast as projection. Instead of finding one intelligent guess for the starting point for searching, Multiple EM elicitation (MEME) \cite{Bailey2006} starts at multiple positions. 
%Other results seek nature-inspired approaches in solving hard problems. Genetic algorithms (GA) are some of the most commonly used because they are easily implementable on a computer. Given a good representation of the solution space and strategies to avoid being trapped on local optimal solutions, GA often converges to an acceptable solution. The collection of genetic algorithm-based approaches for finding CSP includes \cite{Liu2004},  \cite{Dinu2012} and \cite{Clemente2013}. The quality of the solution produced by these algorithms is usually measured using empirical data. Since the quality of the solution is based solely on probability, probabilistic algorithms are expected to run a number of times before getting an acceptable solution.
When it is possible to guarantee the quality of the solution by means of identifying the bounds of its solution's cost, it is called an \textit{approximation algorithm}. Since we do not know the cost of an optimal solution, certain properties of the input instances and the problem itself are needed to design the algorithm. For the CSP, the first constant-factor $2$-approximation algorithm was presented in \cite{Lanctot1998}, which was subsequently  improved in \cite{Li1999}.

Due to the hardness results presented in \cite{Bodlaender1995}, several other efforts were made to identify  for which types of input instances the problem becomes easier to solve. A result in \cite{Brejova2006} shows that even if the set of input instances is defined over the binary alphabet, we still cannot obtain a practical polynomial-time algorithm for small error bounds.  Aside from characterizing input instances, one line of research focused on the \textit{parameterization} of the problem  \cite{Evans2003}. Based on these studies, it is shown that the problem is \textit{fixed-parameter intractable}, i.e., fixing a parameter such as the pattern length or alphabet size will not make the problem easier to solve.

A PTAS for a hard problem is set of polynomial-time algorithms  for which one can specify a certain guaranteed quality. However, since there is a trade-off between the quality of solution and the  running time of the algorithm, we can expect a longer running time for smaller error bounds. 

The summary of the approximation ratio and running time of the existing PTASs for CSP is shown in Table \ref{table:summary}. The second on the list is a randomized PTAS for CSP, while the third one assumes a general alphabet $\Sigma$. We can see that the approximation ratio of the third PTAS is only dependent on the sampling size $r$, because the alphabet size is compensated in the running time of the algorithm. In contrast, we have the alphabet size as a parameter in the approximation ratio of the first PTAS. 

\begin{table}
\centering
\resizebox{0.495\textwidth}{!}{  
\def\arraystretch{2}
\begin{tabular}{| c | c | c |}
\hline
\textbf{PTAS} & \textbf{Approximation Ratio} & \textbf{Running Time} \\ \hline
\textbf{[Li1999]} & $1 + \frac{4|\Sigma| - 4}{\sqrt{e} (\sqrt{4r+1}-3)} $ & $O(l(tn)^{r+1})$\\ \hline
\textbf{[Li1999]} & $1 + O(\sqrt{\frac{\log{r}}{r}})$ & $O(l(tn)^{r+1})$ \\ \hline
\textbf{[Ma2000]} & $1 + \frac{1}{2r-1} + 3 \epsilon r $ & $O(l (tn)^{r+1} |\Sigma|^{\sqrt{4/\epsilon^2} \log{tn}})$\footnote{where $0 < \epsilon \leq 1, {e}/{(1+\epsilon)^{1 + \frac{1}{\epsilon}}} \leq exp(-\frac{\epsilon}{3})$} \\ \hline 

%\textbf{Algorithm \ref{S_alg}} &  $\frac{(2\sigma -1)}{\sigma} + \frac{(l-1)(\sigma-1)}{\sigma}$ & $O((t+1) \cdot (Time(ALG) + l\cdot n))$ \\ \hline
%
%\textbf{Algorithm \ref{S_alg_k}} & $\frac{(2\sigma -1)}{\sigma} + \frac{k\cdot(l-1)(\sigma-1)}{\sigma}$   & $O({(t+k) \choose k } \cdot (Time(ALG) + kl\cdot n))$ \\ \hline
%
%\textbf{Algorithm \ref{algo:iterative}} &  $  \frac{(2\sigma -1)^k}{\sigma^k} + O(\frac{k\cdot(l-1)(\sigma-1)}{\sigma}) $ & $O((k^2 + tk)\cdot (Time(ALG)+ l\cdot n))$ \\ \hline
\end{tabular} }
\caption{Summary of approximation ratio and running time of %algorithms for the reoptimization variants of the CSP together with 
the existing PTASs for CSP in the literature.}
\label{table:summary}
\end{table}

We consider the first PTAS in Table \ref{table:summary} for our study, since we will  only focus on deterministic algorithms for reoptimization. Moreover, we assume that the set of input instances are obtained from  a general alphabet $\Sigma$ and so choosing the third PTAS may not be ideal in our case since it will not just have $r$ as a parameter but also the alphabet size.% 

The first PTAS from \cite{Li1999} is shown in Algorithm \ref{alg:ptas_csp}. For each parameter $r$, it describes an approximation algorithm for CSP that outputs a solution SOL with  $$cost(SOL) \leq \left( 1 + \frac{4|\Sigma| - 4}{\sqrt{e}\sqrt{4r+1}-3} \right) \cdot cost(OPT)$$  in $O(l(tn)^{r+1})$ time.

An $r$-\textit{sample} from a given instance $\S$, i.e., 
$$r\text{-}sample(\S)\ = \{y_{i_1},  y_{i_2}, \ldots, y_{i_r}\},$$ is a collection of $r$ $l$-length substrings  from $\mathcal{S}$.  Repetition of substrings are  allowed for as long as no two substrings are obtained from the same sequence.  Let $R(\mathcal{S})$ denote the set of all possible  $r$-\textit{sample} from $\mathcal{S}$.  The total number of 
samples in $\mathcal{S}$ is  ${tn \choose r}$ which is bounded above by $O((tn)^r)$.  From each sample, the algorithm obtains a consensus pattern. Solution SOL is then derived by aligning the $t$ closest substrings from the given consensus. The Algorithm \ref{alg:ptas_csp} minimizes  through all possible $r$-\textit{sample} in \S  to provide a feasible solution with a guaranteed quality.

\begin{algorithm}
\textbf{Input:} {Set of sequences $\mathcal{S}$, {pattern length} $l$, {sampling size $r$}} \\
\textbf{Output:} {SOL and consensus string $v_{sol}$} 
\begin{algorithmic}[1]

\State{$min = \infty$}
\For{\textbf{each} $r$-$sample$ $\{y_{i_1}, \ldots, y_{i_r}\}$ from $\mathcal{S}$}
\State{$v =$ consensus pattern from $\{y_{i_1}, \ldots, y_{i_r}\}$}
\State{$SOL = \emptyset$ }
\For{\textbf{each} $S_i \in \mathcal{S}$}
\State{$y_i =$  $\min\limits_{\forall y_i \in S_i} d(v, y_i)$  }
\State{$SOL =$ merge $y_i$ to SOL }
\EndFor
\State{\textbf{return} SOL with $\min cost(SOL)$}
\EndFor
\end{algorithmic}
%\begin{enumerate}
%\item For every $r$ $l$-length substrings obtained from distinct sequences in $\mathcal{S}$.
%\begin{enumerate}
%\item Compute the consensus string $v_{r}$ from $r$ substrings.
%\item For each $S_i \in \mathcal{S}$, get substring $y_i$ that is closest to $v_{r}$.
%\item Compute $cost(SOL)$
%\end{enumerate}
%\item Output SOL and corresponding $v_{sol}$ with minimum $cost(SOL)$
%\end{enumerate}
\caption{PTAS for the CSP [Li1999]}
\label{alg:ptas_csp}
\end{algorithm}

\section{Self-reducibility}  \label{sec:selfreducibility}
A solution to a combinatorial optimization problem is composed of a set of discrete elements called \textit{atoms} \cite{Zych2012}. For certain graph problems, it can be the set of all vertices or set of all edges, e.g., a clique in the maximum clique problem is a collection of vertices.

A problem $\Pi$ is said to be \textit{self-reducible} if there is a polynomial-time algorithm, $\Delta$, satisfying the following conditions \cite{Zych2012}.
\begin{enumerate}
\item Given an instance $I$ and an atom $\alpha$ of a solution to $I$, $\Delta$ outputs an instance $I_\alpha$. We require that the size of $I_\alpha$ is smaller than the size of $I$, i.e., $|I_\alpha| < |I|$. Let $\mathcal{R}(I|\alpha) $ represent the set of feasible solutions to $I$ containing  the atom $\alpha$. We require that every solution SOL of $I_\alpha$, i.e., $SOL \in \mathcal{R}(I_\alpha)$, has a corresponding $SOL \cup \{\alpha\} \in \mathcal{R}(I|\alpha)$ and that this correspondence is one-to-one.
\item For any set $H \in \mathcal{R}(I_\alpha)$ it holds that the 
$$cost(I, H \cup  \{\alpha\}) =   cost(I_\alpha, H)+ cost(I, \alpha).$$
\end{enumerate}

Given the properties of a self-reducible problem, we prove that the following lemma is true. 

\begin{lemma}
CSP is \textit{self-reducible}.
\end{lemma}

\begin{proof}
With the assumption that the pattern length $l$ is constant, a valid input instance $I$ to CSP is a set of $t$ sequences $\mathcal{S}$. A feasible solution $SOL \in \mathcal{R}(\mathcal{S})$ is an ordered set  $SOL = (y_1, y_2, \ldots, y_t)$. The set of atoms in CSP are all possible $l$-length substrings in $\mathcal{S}$.
Let us define a reduction function $\Delta(\mathcal{S}, y_i)$, which returns a reduced instance $\mathcal{S} \setminus \{S_i\}$.
The reduced instance is derived by removing one sequence $S_i$  where $y_i$ is obtained. We argue next that $\Delta(\mathcal{S}, y_i)$ has the following properties.

\begin{enumerate}
\item For a $SOL \in \mathcal{R}(\mathcal{S})$ there is a corresponding  $SOL_\alpha \cup \{\alpha\} \in \mathcal{R}(I| \alpha)$. Note that,  $SOL = (y_1, y_2, \ldots, y_t)$ guarantees at least one occurrence of the closest substring per sequence. For any atom $y_i$, a solution $SOL \in \mathcal{R}(\mathcal{S}|y_i)$ corresponds to $SOL_{i} \oplus  y_i $, i.e.,
 $$(y_1,y_2, \ldots, y_{i-1}, y_{i+1}, \ldots, y_t) \oplus y_i \in \mathcal{R}(\mathcal{S} | y_i ).$$
Instead of using the operation 
`$\cup$' for sets, we used $\oplus$ to denote the merge of an element to a sequence.
\item For any feasible solution $SOL_i$ to the reduced instance $\mathcal{S} \setminus \{S_i\}$, we can obtain a feasible solution $SOL_i\oplus y_i$ for $\mathcal{S}$ with $cost(\mathcal{S}, SOL_i\oplus y_i) $  equal to the sum of $ cost(\mathcal{S}\setminus \{S_i\}, SOL_i) $ and $cost(\mathcal{S}, y_i) $, since\\

\resizebox{0.46\textwidth}{!}{
\begin{tabular}{l l}
$cost(SOL_i\oplus y_i)$ & $= cost(\mathcal{S}\setminus \{S_i\}, SOL_i) + cost(\mathcal{S}, y_i)$ \\
 %&  $= \sum\limits_{\substack{j=1\\ j\neq i}}^t \sum\limits_{\substack{k=j+1 \\ k\neq i}}^t d(y_j, y_k)  + \sum\limits_{\substack{j=1 \\j \neq i}}^t d(y_i, y_j).$\\ 
& $ =  \sum\limits_{\substack{i=1\\ j \neq i}}^t d(v, y_j)  +  d(v, y_i)$,\\
\end{tabular}
}
where $v$ is the consensus of $SOL_i\oplus y_i$.
\end{enumerate}
\end{proof}

Let us use the concept of polynomial-time reduction for combinatorial problems. We say that a problem $A$ is polynomial-time reducible to $B$, denoted by $A \leq_P B$, if $\exists $ a polynomial-time transformation $f$, which for every input 
$$x \in A \leftrightarrow f(x) \in B.$$
In other words, in order to solve problem $A$, we must at least solve problem $B$. The problem of finding the closest substring is reduced to a graph problem in \cite{Pevzner2000}. The transformation is as follows. Given a set of sequences $\mathcal{S}$ and a pattern length $l$, an edge weighted  $t$-partite graph $G_\mathcal{S} = (V, E, c)$ is obtained, where the problem is reduced to finding  a minimum weighted clique on a $t$-partite graph (MWCP). Each substring in $\mathcal{S}$ represents a vertex in $G_\mathcal{S}$. A part $V_i \subset V$ represents the set of vertices obtained from a single sequence $S_i \in \mathcal{S}$. A vertex $v \in V_i$ is connected to all other vertices in $V$ except those belonging to $V_i$. The cost defined by function $c:(V \times V) \rightarrow \mathbb{Z}^+$ between two vertices can be interpreted as the Hamming distance between two substrings in $\mathcal{S}$. The cost of a clique in $G_\mathcal{S}$, which is computed by getting the sum of all the edges is equal to sum of all pairwise Hamming distances of substrings in SOL in $\mathcal{S}$. Moreover, we can create an instance $G_\mathcal{S}$ from $\mathcal{S}$ in polynomial-time. Thus, showing a polynomial-time reduction from CSP to MWCP. 

It is also shown that MWCP has an exact reduction to  Minimum Weighted Independent Set Problem (MWISP), since a clique in a graph is an independent set in the corresponding complement of the graph \cite{Garey1979}. In line with this we would like to cite the following Lemma from \cite{Zych2012}. 

\begin{lemma}
Maximum Weighted Independent Set Problem (MWISP) is \textit{self-reducible}.
\end{lemma}

Since we proved earlier that CSP is self-reducible and we know that there is a polynomial-time reduction from CSP to MWCP, which is equivalent to a self-reducible problem MWISP, we are interested to know if all problems that are polynomial-time reducible to self-reducible problems also exhibit the same property.

\begin{theorem} \label{thm:gen_selfreduc}
If problem $A$ is polynomial-time reducible to problem $B$ ($A \leq_{P} B$), and $B$ is self-reducible, then  $A$ is self-reducible. 
\end{theorem}

\begin{proof}
Let $A = (\mathcal{D}_A, \mathcal{R}_A, cost_A, goal_A)$ and \\
$B = (\mathcal{D}_B, \mathcal{R}_B, cost_B, goal_B)$ be two NP optimization problems. Let $I_A \in \mathcal{D}_A$, $SOL_A \in \mathcal{R}(I_A)$, where $SOL_A$ is composed of atoms $\alpha_A$. Similarly, let $I_B \in \mathcal{D}_B$, $SOL_B \in \mathcal{R}(I_B)$, where $SOL_B$ is composed of atoms $\alpha_B$.

Given that $A \leq_P B$, then by definition, there exists a polynomial-time computable function $f$ such that  for every instance $I_A$ for problem $A$, $f(I_A)$ is an instance of problem $B$ and for every solution $SOL_A$ to $A$, $f(SOL_A) \in \mathcal{I_B}$. Equivalently, as a decision problem, the polynomial-time reducibility implies
$$(I_A, SOL_A) \in (\mathcal{D}_A \times \mathcal{R(I_A)}) $$ $$\leftrightarrow $$ $$(f(I_A), f(SOL_A)) \in (\mathcal{D}_B \times \mathcal{R(I_B)})  $$
%Given that $A \leq_P B$, then by definition, we can obtain a polynomial-time computable functions $f_{in}$ and $f_{out}$ such that we can transform an instance $I_A$ through $f_{in}(I_A)$ to get $I_B$  and solve for $SOL_A = f_{out}(SOL_B)$ using the solution of $B$.
By definition of self-reducibility, if $B$ is self-reducible then we can obtain a self-reduction function $\Delta_B(I_B, \alpha_B) = I_{\alpha_B}$, such that $|I_{\alpha_B}| < |I_B|$ and the following conditions hold.
\begin{enumerate}

\item For $SOL_B \in \mathcal{R}(I_B)$ there is a corresponding  $SOL_{\alpha_B} \cup \{\alpha_B\} \in \mathcal{R}(I_B| \alpha_B) $, where $SOL_{\alpha_B} \in \mathcal{R}(I_{\alpha_B})$.

\item For a subset of atoms $H_B \subseteq \mathcal{R}(I_{\alpha_B})$,   $cost_B(I_B, H_B \cup  \{\alpha_B\}) = cost_B(I_B, \alpha_B) + cost_B(I_{\alpha_B}, H_B)$.  
\end{enumerate}

Given that $B$ is self-reducible, we need to show that problem $A$ is also self-reducible.
If so, we must construct a self-reduction function $\Delta_A(I_A, \alpha_A) = I_{\alpha_A}$ that follows the two properties.
 
Given the polynomial-time function $f$ and the self-reduction function $\Delta_B$, we realize $\Delta_A$ using $\Delta_B$ through the following,
 $$\Delta_B(f(I_A), f(\alpha_A)) = f(I_{\alpha_A}).$$ 
The self-reduction function for $\Delta_A$ clearly inherits the two conditions because of our premise that $B$ is self-reducible. Moreover, the reduction function runs in polynomial-time since $f$ is polynomial-time computable.
 \end{proof}
 
Theorem \ref{thm:gen_selfreduc} is presented in the hope that we can use the current approaches for the reoptimization variants of  clique and independent set problem for providing improvements for CSP. It is shown that for some defined reoptimization variant of clique and independent set, a general method is shown to improve the approximability with trade-off on the running time of the approximation algorithm.

\section{Reoptimization}

For real-world applications, additional information about the problems we are solving often is available and so we may not have to solve them from scratch. One of the approaches is to make use of a priori information, which can be a solution to a smaller input instance of a problem to solve a larger instance of it. This approach is called \textit{reoptimization}. The idea was first mentioned in \cite{Schaffter1997}. For some problems, we can transform the given optimal solution so that it may become feasible for the modified instance in polynomial-time. Furthermore, this approach can help to improve the approximability of the problem or the running time of the algorithms solving it. In fact, we can obtain a PTAS for a reoptimization variant of some problem given that the unmodified problem has a constant-factor approximation algorithm \cite{Bockenhauer2008}. 

In the reoptimization variant of any problem, it is important to define precisely the modification relation $\mathcal{M}$ over the set of input instances.  For graph problems, common modifications involve addition/deletion of edges and vertices. Other types of modification involve changes in the edge/vertex weights. A simple definition of reoptimization is as follows.\\

\noindent INPUT: Original instance $I$, its optimal solution OPT, and a modified instance $I^\prime$, where $(I,I^\prime) \in \mathcal{M}$ \\
\noindent OUTPUT: Solution $SOL^\prime$ to $I^\prime$.\\

For the CSP, we consider the basic type of modification where one  or several sequences are added to $\mathcal{S}$. Since the length of the pattern remains unchanged, we will only define the modification relation over the given set of sequences. When a single sequence is added to $\mathcal{S}$, we have the modification relation $ \mathcal{M}_+$, where $(\mathcal{S},\mathcal{S}^\prime) \in \mathcal{M}_+$, if $\mathcal{S}^\prime= \mathcal{S} \cup \{S_{t+1}\}$ and $S_{t+1} \notin \mathcal{S}$. As a generalization, we define the modification relation $\mathcal{M}_{k^+}$ to denote addition of $k$ sequences to $\mathcal{S}$, i.e., $(\mathcal{S}, \mathcal{S}^\prime) \in \mathcal{M}_{k^+}$, if $\mathcal{S}^\prime = \mathcal{S} \cup \{S_{t+1}, S_{t+2}, \ldots, S_{t+k}\}$ and $S_{t+i} \notin \mathcal{S}$, for $1 \leq i \leq k$.

Let us define the reoptimization variant of the CSP under single sequence addition. In the following definition, we can see that we already have the optimal solution with closest substring pattern $v_{opt}$, of the original instance $\mathcal{S}$. Note that we can easily compute occurrences of $v_{opt}$ in $\mathcal{S}$, i.e., the set of substrings $y_i$, each from $S_i$ and their positions where $v_{opt}$ is obtained.
\begin{definition}{$Reopt$-$CSP_{\mathcal{M}_+}$} \ \\ %$(\mathbf{R}_\mathcal{M}(CSP))$] \label{def:CSP_reopt}\ \\
INPUT: Pattern length $l$, original instance $\mathcal{S}$, the optimal closest substring $v_{opt}$ of $\mathcal{S}$, and a modified instance $\mathcal{S}^\prime$ where $(\mathcal{S}, \mathcal{S}^\prime) \in  \mathcal{M}_+$\\
\noindent OUTPUT: Solution $v^\prime_{sol}$ to the modified instance $\mathcal{S}^\prime$.
\end{definition}
Even with the additional information we have, the reoptimization variant of the problem is still NP-hard, i.e., no polynomial-time algorithm exists to obtain the optimal solution for $\mathcal{S}^\prime$, unless $P=NP$. 
\begin{theorem} 
$Reopt$-$CSP_{\mathcal{M}_+}$ is NP-hard.
\label{thm:hard}
\end{theorem}
\begin{proof}
Towards contradiction, assume that $Reopt$-$CSP_{\mathcal{M}_+}$ is polynomial-time solvable. We will make use of the polynomial-time algorithm for $Reopt$-$CSP_{\mathcal{M}_+}$ to solve the CSP. 

We start by showing that for two sequences, both of length $n$, in $\mathcal{S}$, we can obtain the optimal closest substring $v_{opt}$ in polynomial-time. We can obtain $v_{opt}$ in $O(l \cdot n^2)$ by exhausting all possible $l$-length patterns in both sequences.

By making use of the optimal solution $v_{opt}$ and a polynomial-time algorithm for $Reopt$-$CSP_{\mathcal{M}_+}$, we can solve the CSP for any number of sequences in $\mathcal{S}$ in polynomial-time by adding one sequence at a time. But then, we know that closest substring is NP-hard. Therefore, $Reopt$-$CSP_{\mathcal{M}_+}$ is also NP-hard.

\end{proof}

\subsection{Approximation Algorithms for CSP}

For the purpose of our discussion, we may refer to OPT and $OPT^\prime$ to be the optimal solution of the smaller instance and the larger instance of the problem, respectively. The output of the presented approximation algorithms is denoted by $SOL^\prime$, unless otherwise stated. 

First, we show how we can use a  simple algorithm to give an approximate solution for $Reopt$-$CSP_{\mathcal{M}_+}$. The algorithm  makes use of the given solution as a greedy partial solution to the larger instance. 

\begin{algorithm} 
\caption{Given a sequence $S_{t+1} \in \Sigma^n$ and solution SOL of input instance $\mathcal{S}$, procedure BEST-ALIGN produces a feasible solution $SOL^\prime$ by aligning the closest substring from $S_{t+1}$ to $v_{sol}$}
\label{alg:best_align}
\begin{algorithmic}[1]
\Procedure{BEST-ALIGN}{$SOL, S_{t+1}$}
\State{$v_{sol} =$ Consensus pattern from SOL}
\State{$min = \infty$ }
\For{\textbf{each} $l$-length substring $x$ of $S_{t+1}$} 
\If{$d(v_{sol},x) < \ min$ }
\State{$min = cost(x)$}
\State{$y_{t+1} = x$}
\EndIf
\EndFor
%\State{$SOL^\prime = SOL \cup \{y_{t+1}\}$}
\State{$SOL^\prime = (y_1, \ldots, y_t,  y_{t+1})$}
\State{\textbf{return} $SOL^\prime$}
\EndProcedure
\end{algorithmic}
\end{algorithm}

Algorithm \ref{alg:best_align} searches for the best aligning substring $y_{t+1}$ to the given solution. In adding one sequence to $\mathcal{S}$, the computed solution may or may not be the optimal solution for the larger instance. %For the first case, if $v_{opt}$ is part of the optimal solution for $\mathcal{S}^\prime$, then the result of Algorithm \ref{extensions} will yield the optimal solution for $\mathcal{S}^\prime$. 

For the first case, if OPT is subset of $OPT^\prime$ for $\mathcal{S}^\prime$, then the result of Algorithm \ref{alg:best_align} will yield the optimal solution for $\mathcal{S}^\prime$. Otherwise, there exists a non optimal solution for $\mathcal{S}$ that is part of the optimal solution for $\mathcal{S}^\prime$, i.e., $\exists SOL \subset  OPT^\prime$. The solution $SOL^\prime$ of Algorithm \ref{alg:best_align} is obtained by merging the given optimal solution OPT with the best possible aligning substring in the new sequence. Algorithm \ref{alg:best_align} obviously runs in linear time with respect to the length of the additional sequence $ S_{t+1}$. 
To illustrate that we cannot always get the optimal solution using Algorithm \ref{alg:best_align}, let us consider the following example.

\begin{example} \label{ex:1}
Let $\mathcal{S} = \{S_1, S_2, S_3, S_4 \}$ with $S_5$ as the additional sequence for the modified input $\mathcal{S}^\prime$. If we are looking for a closest substring of length $l = 4$, the optimal solution for $\mathcal{S}$ is  $v_{opt}: AAAA $.  

\begin{center}
\begin{tabular}{ r  c c c c c c c c}
$S_1$: &  {\bf A} & {\bf A} & {\bf A} & {\bf A} & B & B & B & B \\
$S_2:$ & B & B & B & B &  {\bf A} & {\bf A} & {\bf A} & {\bf A}  \\
$S_3: $ & {\bf A} & {\bf A} & {\bf A} & {\bf A} & B & B & B & A \\
$S_4:$ & B & B & B & B &  {\bf A} & {\bf A} & {\bf A} & {\bf A}  \\ \\

$S_5: $ &  {\bf B} & {\bf B} & {\bf B} & {\bf B} & B & B & B & B \\
\end{tabular}
\end{center}

Algorithm \ref{alg:best_align} will return a solution by aligning the most similar substring $y_5: BBBB$ from the new sequence $S_5$. The solution produced by Algorithm \ref{alg:best_align} will have $cost(SOL^\prime) = 4$. On the other hand,  the optimal solution for $\mathcal{S}^\prime$ is $v^\prime_{opt}: BBBB$ with $cost(OPT^\prime) = 1$. 
\end{example}

In this case, the subset of the optimal solution for the larger instance is not the optimal solution for the smaller instance.  Algorithm \ref{alg:best_align} obviously runs in linear time with respect to the length of the additional sequence $ S_{t+1}$, multiplied by getting the cost of each substring $x$. %Since we are given the optimal solution for $\mathcal{S}$ which is  $$OPT = (v_{opt}, y_1, y_2, \ldots, y_t),$$ where $v_{opt}$ is the consensus pattern of substrings  $y_1, \ldots, y_t$, we have the following solution for $\mathcal{S}^\prime$.

%$$SOL^\prime = (v_{sol^\prime}, y_1, y_2, \ldots, y_t, y_{t+1})$$

 %Moreover, the set of  $v^\prime_{sol}$'s may not necessarily be equal to the set of $v_{opt}$'s, because now we considered $y_{t+1}$ to be part of the solution. For instance, if we have the substrings $(y_1,y_2, y_3)$ with one consensus pattern $v_{opt}:$ABC, after considering  $y_4:$BCD, the alignment would give us another closest substring $v_{sol^\prime}:$BCD,  which can be totally different from $v_{opt}$.
%
%\begin{center}
%\begin{tabular}{r r r r  }
%$y_1:$&B&B&C\\
%$y_2:$&A&C&C\\
%$y_3:$&A&B&D\\ \ \\
%$y_4:$& B&C&D \\
%\hline
%$v_{sol^\prime}:$&B&C&D\\ 
%\end{tabular}
%\end{center}

To get the quality of the approximation algorithm, we need to compare $cost(OPT^\prime)$ and $cost(SOL^\prime)$. Using Algorithm \ref{alg:best_align} and the second property of self-reducibility, we have
$$cost(SOL^\prime) = cost(OPT) + d(v_{opt}, y_{t+1}),$$ % \leq cost(OPT) + l .$$ 
where $y_{t+1}$ is the best aligning substring to $v_{opt}$ from the new sequence. We know that the cost of the optimal solution for the smaller instance is less than or equal to the cost of the optimal solution for the larger instance, i.e., $cost(OPT) \leq cost(OPT')$. Therefore,
$$cost(SOL^\prime)  \leq cost(OPT^\prime)+ d(v_{opt}, y_{t+1}).$$
In the worst case scenario, we can get $$cost(SOL^\prime) \leq cost(OPT^\prime) + l .$$ 
Since we only added a single sequence, the quality of the solution depends solely on how long the pattern is.  Using Algorithm \ref{alg:best_align}, we can get  a feasible solution for $\mathcal{S^\prime}$ in $O(ln)$. Compared to the $2$-approximation algorithm for CSP with running time $O(l(tn)^2)$ from \cite{Lanctot1998}, our approach can benefit an improved approximation ratio and running time for instances $\mathcal{S}^\prime$ where $OPT^\prime < l$.
Note that, the first property of the self-reducibility of CSP allows us to provide a feasible solution for the modified instance by extending the additional information that we have. We illustrate in the following generalization that it is possible to produce a feasible solution even if we add $k$-sequences to $\mathcal{S}$. 
 
%But can we further improve the approximation ratio? Suppose we have an existing $\sigma$-approximation algorithm $ALG$, where $\sigma \geq 1$, and we  use it to provide another approximate solutions for other smaller instances of $\mathcal{S}^\prime$. We may choose the best solution between the two approximation algorithms. 

%Since we have only added a single sequence, the quality of the solution depends on how long the pattern is and the number of sequence of the original instance, in this case $t$. But can we further improve the running time of the PTAS for CSP?

\subsection{Generalization}
In this section, we consider the extension of the reoptimization variant where, instead of adding one sequence, the modification relation, denoted by $\mathcal{M}_k^+ $, is characterized by adding $k$ sequences to the original instance. The definition of the generalized reoptimization version under sequence addition is as follows.
\begin{definition}{Reopt-CSP$_{\mathcal{M}_k^+}$} \ \\ %$(\mathbf{R}_\mathcal{M}(CSP))$] \label{def:CSP_reopt}\ \\
INPUT: Pattern length $l$, original instance $\mathcal{S}$, the optimal closest substring $v_{opt}$ of $\mathcal{S}$, and a modified instance $\mathcal{S}^\prime$ where $(\mathcal{S}, \mathcal{S}^\prime) \in  \mathcal{M}_k^+$\\
\noindent OUTPUT: Solution $v^\prime_{sol}$ to the modified instance $\mathcal{S}^\prime$.
\end{definition}
Since $Reopt$-$CSP_{\mathcal{M}_k^+}$ is a generalization of $Reopt$-$CSP_{\mathcal{M}_+}$, where $k=1$, we no longer need to show that this variant is also NP-hard. To give an approximate solution, we  can generalize Algorithm \ref{alg:best_align} for $Reopt$-$CSP_{\mathcal{M}_k^+}$. 

\begin{algorithm} 
\caption{Given $k$ additional sequences $\{S_{t+1}, \ldots S_{t+k}\}$, where each $ S_{t+i} \in \Sigma^n$ and solution SOL for input instance $\mathcal{S}$, procedure K-BEST-ALIGN produces a feasible solution $SOL^\prime$ by aligning the closest substrings from each of the additional sequences to $v_{sol}$}\label{alg:k_best_align}
\begin{algorithmic}[1]
\Procedure{K-BEST-ALIGN}{$SOL, \{S_{t+1}, \ldots S_{t+k}\}$}
\State{$v_{sol} =$ Consensus pattern from SOL}
\For{\textbf{each} $i$   in $1$ to $k$}
\State{$min = \infty$ }
\For{\textbf{each} $l$-length substring $x$ of $S_{t+i}$} 
\If{$d(v_{sol},x) < \ min$ }
\State{$min = cost(x)$}
\State{$y_{t+i} = x$}
\EndIf
\EndFor
\EndFor
%\State{$SOL^\prime = SOL \cup \{y_{t+1}\}$}
\State{$SOL^\prime = (y_1, \ldots, y_t,  y_{t+1}, \ldots, y_{t+k})$}
\State{\textbf{return} $SOL^\prime$}
\EndProcedure
\end{algorithmic}
\end{algorithm}

Given $v_{opt}$, we can give a feasible solution for $\mathcal{S}^\prime$ by getting the best aligning substring $y_{t+i}$ from each of the newly added sequences. Let $Y$ be the set of substrings $y_{t+i}$ and let 
$$cost(Y) = \sum\limits_{i=1}^k d(v_{opt}, y_{t+i})$$
be the contribution of the set $Y$ to $SOL^\prime$. If $v_{opt}$ is the same as $v^\prime_{opt}$, then we can always guarantee optimality, otherwise
$$cost(SOL^\prime) \leq cost(OPT^\prime) + cost(Y).$$
In the worst case scenario, we can have
$$cost(SOL^\prime) \leq cost(OPT^\prime) +kl.$$
The extension of Algorithm \ref{alg:k_best_align} for $Reopt$-$CSP_{\mathcal{M}_k^+}$ can produce a feasible solution in $O(kln)$. However, the quality of the solution degrades as we increase the number of sequences added to the original instance. A general method for reoptimization is applied to several self-reducible hard problems to further improve the approximability of the problem but with tradeoff on the running time of the algorithm. The PTAS for the CSP already provides the option of improving the approximation ratio by increasing  the sampling size $r$ in Algorithm \ref{alg:ptas_csp}. As $r$ approaches $t$,  we can have a solution with error that converges to $0$ but with running time that is comparable to the naive exhaustive search that is exponential in $t$. In the following section, we illustrate how reoptimization can help in improving the running time of the PTAS in \cite{Li1999}.

\subsection{Improving the PTAS for Reopt-CSP$_{\mathcal{M}_{k^+}}$}

Let us consider an input instance $\mathcal{S^\prime} = \{S_1, \ldots, S_t \}$ for CSP with a given optimal solution for a subset of its sequences. Without loss of generality, suppose we have an optimal solution OPT for the first $r$ sequences in $\mathcal{S}^\prime$, i.e., $\mathcal{S} = \{S_1, \ldots, S_r\}$. Let us also assume that $|\mathcal{S^\prime} \setminus \mathcal{S}| =k$.% 

Algorithm \ref{alg:ptas_csp} can provide a feasible solution for \Sp with an approximation ratio of  $1 + \frac{4|\Sigma| - 4}{\sqrt{e} (\sqrt{4r+1}-3)} $ in $O(l\cdot(tn)^{r+1})$. We argue in this section that we can achieve the same approximation ratio in $O(ltn((t-r)n)^{r})$ time using the given OPT for $\mathcal{S}$, as presented in Theorem \ref{thm:imptas}

In the following approximation algorithm, we implement the PTAS in Algorithm \ref{alg:ptas_csp} by using the given optimal solution OPT.

\begin{algorithm}
\caption{Approximation algorithm for Reopt-$CSP_{\mathcal{M}_{k^+}}$ using the K-BEST-ALIGN procedure from Algorithm \ref{alg:k_best_align}.}
\label{alg:imptas}
%$ALG$setup{indent=2em}
\noindent \textbf{Input: } Set of sequences $\mathcal{S}^\prime = \{S_1, \ldots, S_t\}$, pattern length $l$, and optimal closest substring OPT of $\mathcal{S} = \{S_1, \ldots, S_r\}$ \\
\textbf{Output:}  Solution $SOL^\prime$ for $\mathcal{S}^\prime$.\\

\begin{algorithmic}[1]
\State{$SOL_A^\prime = \text{K-BEST-ALIGN}(OPT, \{S_{r+1}, \ldots, S_{t}\})$} \\

\State{$SOL_B^\prime = \emptyset $}
\State{$min= \infty$}
\For{\textbf{each} $r$-$sample \in \{ R(\mathcal{S}^\prime) \setminus R(\mathcal{S})\}$ }
%\State{$SOL_i = ALG(\mathcal{S}^\prime \setminus \{S_{i_1}, \ldots, S_{i_k}\})$ }
%\State{$v =$ Consensus pattern from $sample$}
\State{$SOL^\prime=$ K-BEST-ALIGN$(SOL, sample)$}

\If{$cost(SOL^\prime) < min$ }
\State{$min = cost(SOL^\prime)$}
\State{$SOL_B^\prime = SOL^\prime$}
\EndIf
\EndFor \\

\If{$cost(SOL^\prime_A) < cost(SOL_B^\prime)$}
\State{\textbf{return} $SOL_A^\prime$}
\Else{}
\State{\textbf{return} $SOL_B^\prime$}
\EndIf
\end{algorithmic} 
\end{algorithm}

\begin{theorem}
\label{thm:imptas}
Algorithm \ref{alg:imptas} is a $1 + \frac{4|\Sigma| - 4}{\sqrt{e} (\sqrt{4r+1}-3)} $ PTAS for  Reopt-CSP$_{\mathcal{M}_k^+}$  which runs in $O(ltn\cdot (t-r)n)^{r})$.
\end{theorem} \vspace{3pt}

\begin{proof}
To prove that Algorithm \ref{alg:imptas} provides the same approximation ratio as Algorithm \ref{alg:ptas_csp}  from \cite{Li1999}, we need to show that the set of all samples in $R(\Sp)$ is considered in Algorithm \ref{alg:imptas}. $SOL_A$ is obtained by aligning  $OPT$ to $k$ other sequences in \Sp using Algorithm \ref{alg:best_align}. Since we have the optimal solution for \S, $cost(SOL_A)$  is equal to the minimum over all feasible solutions sampled over $R(\S)$. On the other hand,  solution $SOL_B$ is obtained by getting the minimum over all feasible solutions sampled from $R(\Sp) \setminus R(\S)$. The output $SOL^\prime$ is the minimum between $SOL_A$ and $SOL_B$. Thus, all samples from $R(\Sp)$ are considered   in Algorithm \ref{alg:imptas}.

For the running time of Algorithm \ref{alg:imptas}, getting $SOL_A$ will require $O(kn)$ steps. Line 5 will iterate in $O((tn)^r - (rn)^r)$.  Getting the best alignment for each sample will require $O(tn)$ and getting the distance between two $l$-length substrings will take $O(l)$. Thus, Algorithm \ref{alg:imptas} runs in $O(ltn\cdot (t-r)n)^{r})$.
\end{proof}

\section{Conclusion and Future Work}

In this paper, we presented reoptimization variant of the CSP under sequence addition. The additional information, i.e., optimal solution to a smaller instance, does not help to address the intractability of the problem, as shown in Theorem \ref{thm:hard}.

We considered the basic type of modification $\mathcal{M}_+$, where the new instance contains an additional sequence. The general approach in reoptimization is to transform the given optimal solution for it to become feasible to the new instance. For some cases, i.e., $v_{opt} = v^\prime_{opt}$, our simple transformation can lead to $OPT^\prime$, otherwise we can have the worst case scenario where OPT cannot help to obtain $OPT^\prime$. In fact, even for a small modification such as $\mathcal{M}_+$,  $v_{opt}$ can be totally different from $v^\prime_{opt}$.

We generalized this approach by considering the modification $\mathcal{M}_{k^+}$, characterized by adding $k$ new sequences to the original instance. By using the same approach for Reopt-CSP$_{\mathcal{M}_+}$, we can obtain an error that grows  as the number of additional sequences is increasing. However, we showed that by using the given optimal  solution for reoptimization we can improve the running of the existing PTAS for CSP. With the same approximation ratio, we can improve the running time from $O((tn)^r)$ to $O(((t-r)n)^r)$. 

It is natural to consider the opposite scenario, where an optimal solution to a larger instance is given and we are looking for the solution to a smaller subset of the instance. We argue that we are not also guaranteed  to get the optimal solution by transforming the larger solution to provide a feasible solution to the smaller instance. Example \ref{ex:1} can also illustrate the idea. Exploring other techniques to make use of such kind of information is still part of our future work. Moreover, we would like to study other types of modification in the input instance such as changes in pattern length $l$. 

\section{Acknowledgements}
The authors would like to give their utmost gratitude to Dr. Hans-Joachim B\"{o}ckenhauer and Dr. Dennis Komm for patiently reading and proof-reading this paper, to Prof. Juraj Hromkovic for his supervision and valuable inputs.

Ms. Jhoirene Clemente would like to thank the  Engineering Research and Development for Technology (ERDT) for funding her PhD program in University of the Philippines Diliman and the Bases Conversion and Development Authority (BCDA) project for funding her Sandwich program.

\bibliography{library}
\bibliographystyle{plain}
\end{document}